\documentclass[a4paper,UKenglish,cleveref, autoref, thm-restate]{lipics-v2021}

\hideLIPIcs  

\title{The Compilability Thresholds of 2-CNF to OBDD} 

\titlerunning{The Compilability Thresholds of 2-CNF to OBDD} 

\author{Alexis de Colnet}{Leiden University, Leiden, Netherlands}{a.e.h.de.colnet@liacs.leidenuniv.nl}{https://orcid.org/0000-0002-7517-6735}{}
\author{Alfons Laarman}{Leiden University, Leiden, Netherlands}{a.w.laarman@liacs.leidenuniv.nl}{https://orcid.org/0000-0002-2433-4174}{}
\author{Joon Hyung Lee}{Leiden University, Leiden, Netherlands}{j.h.lee@liacs.leidenuniv.nl}{https://orcid.org/0000-0001-8628-9922}{}

\authorrunning{A. de Colnet, A. Laarman, J. H. Lee} 
\Copyright{Alexis de Colnet, Alfons Laarman, Joon Hyung Lee} 



\category{} 

\relatedversion{} 

\funding{de Colnet and Laarman are supported by the Netherlands Organization for Scientific Research (NWO/OCW) as
part of \emph{Quantum Limits} (project number SUMMIT.1.1016). Lee is supported by the Dutch Research Council (NWO) as part of the project \emph{Boosting the Search for New Quantum Algorithms with AI (BoostQA)}, file number NGF.1623.23.033, under the research programme \emph{Quantum Technologie 2023}.}

\acknowledgements{}

\nolinenumbers 

\EventEditors{Alexey Ignatiev and Stefan Szeider}
\EventNoEds{2}
\EventLongTitle{29th International Conference on Theory and Applications of Satisfiability Testing (SAT 2026)}
\EventShortTitle{SAT 2026}
\EventAcronym{SAT}
\EventYear{2026}
\EventDate{July 20--23, 2026}
\EventLocation{Lisbon, Portugal}
\EventLogo{}
\SeriesVolume{377}
\ArticleNo{13}

\newcommand\defmath[2]{\newcommand#1{\ensuremath{#2}\xspace}}

\defmath{\calA}{\mathcal{A}}
\defmath{\calE}{\mathcal{E}}
\defmath{\calF}{\mathcal{F}}
\defmath{\calG}{\mathcal{G}}
\defmath{\calGr}{\mathcal{G}}
\defmath{\calH}{\mathcal{H}}
\defmath{\calC}{\mathcal{C}}
\defmath{\scrP}{\mathscr{P}}

\defmath{\var}{var}
\defmath{\tw}{tw}
\defmath{\pw}{pw}
\defmath{\mw}{mw}
\defmath{\gates}{gates}
\defmath{\mmw}{mmw}
\defmath{\mimw}{mimw}
\defmath{\leaves}{leaves}
\defmath{\obddSize}{\text{OBDD-size}}
\defmath{\MF}{\mathit{MF}}
\defmath{\sat}{sat}
\defmath{\CNF}{\mathit{CNF}}
\defmath{\Ex}{\mathbb{E}}

\usepackage[linesnumbered,lined,ruled,vlined]{algorithm2e}
\usepackage{tikz}
\usepackage{mathtools}
\usepackage{mathrsfs}
\usepackage{complexity}
\usepackage{xspace}

\usepackage{bbm}
\usepackage{url}

\begin{document}

\maketitle

\begin{abstract}
We prove the existence of two thresholds regarding the compilability of random 2-CNF formulas to OBDDs. The formulas are drawn from $\calF_2(n,\delta n)$, the uniform distribution over all 2-CNFs with $\delta n$ clauses and $n$ variables, with $\delta \geq 0$ a constant. We show that, with high probability, the random 2-CNF admits OBDDs of size polynomial in $n$ if $0 \leq \delta  < 1/2$ or if $\delta > 1$. On the other hand, for $1/2 < \delta  < 1$, with high probability, the random $2$-CNF admits only OBDDs of size exponential in $n$.
It is no coincidence that the two ``compilability thresholds'' are $\delta = 1/2$ and $\delta = 1$. Both are known thresholds for other CNF properties, namely, $\delta = 1$ is the satisfiability threshold for 2-CNF while $\delta = 1/2$ is the treewidth threshold, i.e., the point where the treewidth of the primal graph jumps from constant to linear in $n$ with high probability.
\end{abstract}

\section{Introduction}

BDDs (Binary Decision Diagrams) are a very well-known model for representing Boolean functions. This model is quite natural and has various properties that make it attractive and therefore has been considerably studied for decades. Here we focus on \emph{ordered} BDDs, or OBDDs, which form an important fragment of BDDs thanks to many desirable properties for practice (canonicity, fast \emph{apply}, etc.). Because of their practical appeal, research on the limits of OBDD developed fast, centered around the following question: which functions admit a small OBDD representation? We will talk of \emph{OBDD-size}, i.e., the size of a minimal-size OBDD for the function. In \emph{knowledge compilation} terms, we want to know which functions are (potentially) easy \emph{to compile} into the OBDD language~\cite{DarwicheM02}, in the sense that they admit small representations in this language. Unfortunately, many families of functions have exponential OBDD-size in the number $n$ of variables~\cite{Zak84,BabaiHST87,Wegener88,Gal97,BolligW98}. Through combinatorial arguments similar to those already used by Shannon for switching circuits~\cite{Shannon49}, one can even show that, as $n$ increases, almost all Boolean functions have exponential OBDD-size~\cite[Theorem 2.2.2.]{Wegener00}. But not all functions are encountered in the wild and we want to focus on the OBDD-size of simple functions, in particular functions that can be represented with compact formulas, since these are not uncommon in practice. 
Several classes of small DNF or CNF formulas were shown to have exponential OBDD-size, including formulas as simple as monotone 2-CNFs~\cite{BolligW98,Razgon21}. For specific classes of CNFs, the structure of the formula, generally captured through  parameters of its underlying graphs (primal graph, incidence graph, hypergraph, etc.) explains large OBDD-size; this is for instance the case for monotone CNFs~\cite{amarilli2020connecting,Razgon21} and Tseitin formulas~\cite{ItsyksonRS22,deColnetM23}.  But these results based on the structure do not directly extend to general CNFs. 
Our aim with this paper is to further our understanding of what makes a CNF formula hard to compile into OBDD using \emph{random CNF formula} models.

\subparagraph*{Contributions.} We focus on random 2-CNFs drawn from $\calF_2(n,m)$, that is,  formulas are taken uniformly at random from the set of 2-CNFs that have $m$ clauses over variables $x_1,\dots,x_n$. Our 2-CNFs are \emph{sparse}: we fix a constant $\delta > 0$ and consider the OBDD-size of formulas drawn from  $\calF_2(n,\delta n)$ as $n$ increases. We show the existence of two \emph{compilability thresholds}: $\delta = 1/2$ and $\delta = 1$. While $\delta < 1/2$, the random CNF almost always has polynomial OBDD size. If $\delta > 1$, again the OBDD-size is almost always small (constant actually). On the other hand,
for $\delta$ between $1/2$ and $1$, almost all random 2-CNF have exponential size. This paper proves Theorem~\ref{theorem:compilability_thresholds}. The proofs of intermediate lemmas and corollaries marked with $(\star)$ appear in appendix.

\begin{restatable}[Compilability Thresholds of 2-CNF to OBDD]{theorem}{compilabilityThreshold}\label{theorem:compilability_thresholds}
Let $\delta \geq 0$, $n \in \mathbb{N}$ and $F$ be a random formula following $\calF_2(n,\delta n)$. There are constants $c,d > 0$ such that, 
\begin{enumerate}[(I)]
\item if $\delta < 1/2$, then $\lim_{n \rightarrow \infty} \Pr[\obddSize(F) \leq n^c] = 1$ \label{item:I}
\item if $1/2 < \delta < 1$, then $\lim_{n \rightarrow \infty} \Pr[\obddSize(F) \geq \exp(n^d)] = 1$
\label{item:II}
\item if $\delta > 1$, then $\lim_{n \rightarrow \infty} \Pr[\obddSize(F) = 1] = 1$
\label{item:III}
\end{enumerate}
\end{restatable}
\subparagraph*{Related work.}
Our work builds on a rich body of results on random 2-CNF formulas. 
The satisfiability threshold (Theorem \ref{theorem:sat_threshold_for_2CNF}) was established independently 
by Chvátal and Reed~\cite{ChvatalR92} and Goerdt~\cite{Goerdt96}. 
The precise number of satisfying assignments of random 2-SAT formulas 
was determined by Achlioptas et al.~\cite{2satsol}. The treewidth threshold we rely 
on is due to Lee, Lee, and Oum~\cite{lee2012rank}. On the knowledge 
compilation side,
phase transitions for compilation have been studied experimentally~\cite{GuptaRM20}. However, surprisingly, a theoretical investigation is still lacking and we have to turn to research on OBDD-based proof systems to find some work that considers OBDDs for random CNF, though not in a setting relevant for us~\cite{FriedmanX13}. 
Our paper provides the first theoretical account of the phase transition 
behavior in OBDD compilation, establishing exact thresholds for 
random 2-CNF.

\section{Known Threshold Phenomena for CNF and Graphs}

We review concepts and threshold phenomena related to (random) CNF formulas and (random) graphs. A formula $F$ in conjunctive normal form, or CNF, is a conjunction of \emph{clauses}, i.e., terms of the form $\bigvee_{i} \ell_i$ where each $\ell_i$ is a \emph{literal}, i.e., a Boolean variable or its negation. Clauses are assumed not to repeat literals, nor to contain opposite literals. The size of $F$, denoted by $|F|$, is its number of clauses. The set of variables appearing in $F$ is denoted by $\var(F)$. An assignment $\alpha$ to a set of Boolean variables $X$ is a mapping from $X$ to $\{0,1\}$. The set of assignments to $X$ is denoted by $\{0,1\}^X$. We denote by $\sat(F)$ the set of assignments to $\var(F)$ that satisfy $F$.
$\SAT$ is the set of satisfiable CNF formulas ($\sat(F) \neq \emptyset$). For $k$ a positive integer, a $k$-CNF is a CNF whose clauses all contain exactly $k$ literals. 

\subsection{Random CNF} A random variable $Y$ following a probability distribution $\mathcal{D}$ is denoted by $Y \sim \mathcal{D}$. For $n,k$ positive integers we let $X_n = \{x_1,\dots,x_n\}$ be the set of all variables and $Cl_{k,n}$ be the set of all $2^k\binom{n}{k}$ possible clauses of size $k$ over $X_n$. Let $0 \leq m \leq |Cl_{k,n}|$, then $\calF_k(n,m)$ is a uniform distribution over all $k$-CNF made of $m$ distinct clauses in $Cl_{k,n}$. Following the established convention~\cite{donald2015art,Achlioptas21}, we study random $2$-CNF drawn from $\calF_2(n,\delta n)$ for $\delta$ a constant. It is well-known that $\delta = 1$ is the satisfiability threshold for random 2-CNF following $\calF_2(n,\delta n)$~\cite{ChvatalR92,Goerdt96}. That is, for a fixed $\delta < 1$, as $n$ increases almost all 2-CNF in $\calF_2(n,\delta n)$ are satisfiable. A contrario, for a fixed $\delta > 1$, as $n$ increases almost all 2-CNF in $\calF_2(n,\delta n)$ are unsatisfiable. This threshold is important for this paper, regardless of the fact that deciding satisfiability of a 2-CNF formula is in~\P.

\begin{theorem}[Satisfiability threshold for 2-CNF,~\cite{ChvatalR92,Goerdt96}]\label{theorem:sat_threshold_for_2CNF}
Let $\delta \geq 0$, $n \in \mathbb{N}$ and $F$ be a random formula following $\calF_2(n,\delta n)$. If $\delta < 1$ then $\lim_{n \rightarrow \infty} \Pr[F \text{ is satisfiable}] = 1$. If $\delta > 1$ then $\lim_{n \rightarrow \infty} \Pr[F \text{ is satisfiable}] = 0$.
\end{theorem}

\noindent The behavior at the threshold itself ($\delta = 1$) is ignored (as in~\cite{ChvatalR92,Goerdt96,Achlioptas21}).

\subsection{Random Graphs and Graph Parameters} 
The \emph{primal graph} of CNF $F$ is the graph whose vertices are the variables of $F$ and where two variables are connected by an edge if and only if they appear together in a clause, i.e., the edge set is $\{ \{x,y\} \mid \exists \text{ clause } C \in F, x \in \var(C) \text{ and } y \in \var(C) \text{ and } x \neq y\}$. Note that it is not a multigraph and that, for 2-CNF, each edge can correspond to at most $4$ clauses. 

\emph{Treewidth} and \emph{pathwidth} are well-known graph parameters~\cite{Bodlaender98,SamerS21}. Their definition is not necessary for this paper. We simply recall that the treewidth of a graph $G$ is an integer between $0$ and $|V(G)| - 1$ denoted by $\tw(G)$, that measures how close $G$ is to a forest. Similarly, the pathwidth of $G$ is an integer between $0$ and $|V(G)| - 1$ denoted by $\pw(G)$, that measures how close $G$ is to a disjoint union of paths. It is known that $\mathrm{tw}(G) \le \pw(G) \le O(\tw(G)\log|V(G)|)$~\cite[Corollary 24]{Bodlaender98}.	We call \emph{primal treewidth} and \emph{primal pathwidth} of $F$ the treewidth and pathwidth of its primal graph. For convenience, given a random CNF $F$ drawn from $\calF_2(n,m)$, we denote by $G_F$ the primal graph of $F$ where the missing variables, i.e., those in $X_n \setminus \var(F)$, are added as isolated vertices. Adding isolated vertices to a non-empty graph does not modify its treewidth nor its pathwidth.
	
Given $p \in [0,1]$ and $n \in \mathbb{N}$, $\calGr(n,p)$ is the probability distribution for the random graph over $n$ vertices $V_n = \{v_1,\dots,v_n\}$ where each one of the $\binom{n}{2}$ possible edges is independently present with probability $p$. The expected number of edges in $\calGr(n,p)$ is $p \binom{n}{2}$. Given $m \in \mathbb{N}$, $\calGr(n,m)$ is the probability distribution for the random graph over $n$ vertices composed of $m$ edges chosen uniformly at random and without replacement from all $\binom{n}{2}$ possible edges.\footnote{Using the same letter $\calG$ for $\calG(n,m)$ and $\calG(n,p)$ is standard. The nature of the second parameter (integer or in $[0,1]$) marks the distinction.} For $\delta$ a fixed constant, $\mathcal{G}(n, p = 2\delta/n)$ and 
$\mathcal{G}(n, m = \delta n)$ have the same asymptotic behavior with respect to monotone graph properties as $n \to \infty$. We will refer 
to~\cite[Corollary 1.16]{JansonLR00} to switch from one distribution to the 
other.

A threshold phenomenon occurs at $\delta = 1/2$ for $\tw(G)$ when $G \sim \calGr(n,p = 2\delta/n)$ or $G \sim \calGr(n,m = \delta n)$. When $\delta < 1/2$, a result of Erd\"os and R\'enyi says that, as $n$ increases, every connected component of $G$ has at most one cycle with high probability~\cite{erd6s1960evolution}, and therefore $G$ has treewidth $2$ at most. On the other hand, when $\delta > 1/2$, the treewidth jumps to $\Omega(n)$~\cite{lee2012rank}.

\begin{theorem}[Treewidth threshold for graphs \protect{\cite[Corollary 1.2]{lee2012rank}}]\label{theorem:tw_threshold_for_graphs}
Let $\delta \ge 0$ a fixed constant, $n \in \mathbb{N}$ and $G \sim \calGr(n,p = 2\delta/n)$. If $\delta < 1/2$ then $\lim_{n \rightarrow \infty} \Pr[\tw(G) \geq 3] = 0$. Furthermore, there is a constant $c > 0$ such that, if $\delta > 1/2$ then $\lim_{n \rightarrow \infty} \Pr[\tw(G) \leq cn] = 0$. 
\end{theorem}

Since treewidth is a monotone graph parameter, i.e., if $G$ is a subgraph of $H$ then $\tw(G) \leq \tw(H)$,~\cite[Corollary 1.16]{JansonLR00} tells us that Theorem~\ref{theorem:tw_threshold_for_graphs} also applies to $\calGr(n,m = \delta n)$.

\begin{theorem}\label{theorem:tw_threshold_for_graphs_uniform}
Let $\delta \ge 0$ a fixed constant, $n \in \mathbb{N}$ and $G \sim \calGr(n,\delta n)$. If $\delta < 1/2$ then $\lim_{n \rightarrow \infty} \Pr[\tw(G) \geq 3] = 0$. Furthermore, there is a constant $c > 0$ such that, if $\delta > 1/2$ then $\lim_{n \rightarrow \infty} \Pr[\tw(G) \leq cn] = 0$. 
\end{theorem}

The distribution for the primal graph of a CNF in $\calF_2(n,\delta n)$ is neither $\calGr(n,2\delta/n)$ nor $\calGr(n,\delta n)$ because several clauses can contribute to the same edge in the primal graph. However, one can show that the number of edges corresponding to more than one clause rarely exceeds $\log(n)$ so that, modulo the addition of at most $\log(n)$ random edges, the primal graph of the random CNF is distributed as $\calGr(n,\delta n)$. Thus, the primal treewidth of the random CNF is dominated by the treewidth of a random graph from $\calGr(n,\delta n)$ but rarely differs by more than $O(\log(n))$. Therefore, the constant-to-$\Omega(n)$ threshold of Theorem~\ref{theorem:tw_threshold_for_graphs_uniform} also applies to the primal treewidth of a random CNF in $\calF_2(n,\delta n)$.

\begin{restatable}[$\star$]{lemma}{twThresholdForCNF}\label{lemma:tw_threshold_for_2CNF}
Let $\delta \ge 0$ a fixed constant, $n \in \mathbb{N}$ and $F \sim \calF_2(n,\delta n)$. If $\delta < 1/2$ then $\lim_{n \rightarrow \infty} \Pr[\tw(G_F) \geq 3] = 0$. Furthermore, there is a constant $c > 0$ such that, if $\delta > 1/2$ then $\lim_{n \rightarrow \infty} \Pr[\tw(G_F) \leq cn] = 0$. 
\end{restatable}

\begin{figure}
\centering
\begin{tikzpicture}[xscale=1.2]
\node (0) at (0,-0.5) {$0$};
\node (0) at (3.33,-0.5) {$0.5$};
\node (0) at (6.66,-0.5) {$1$};
\node (0) at (10,-0.5) {$\delta$};
\draw[|-] (0,0) -- (3.33,0);
\draw[|-] (3.33,0) -- (6.66,0);
\draw[|-latex] (6.66,0) -- (10,0);

\draw[dashed] (3.33,0) -- (3.33,2.5);
\draw[dashed] (6.66,0) -- (6.66,2.5);

\draw[thick,blue,latex-latex] (0,1) to node[inner sep=2,midway, fill=white, text=blue] {SAT} (6.66,1);
\draw[thick,red,latex-latex] (6.66,1) to node[inner sep=2,midway, fill=white, text=red] {UNSAT} (10,1);

\draw[thick,blue,latex-latex] (0,2.05) to node[inner sep=2,midway, fill=white, text=blue] {tw $O(1)$} (3.33,2.05);
\draw[thick,red,latex-latex] (3.33,2.05) to node[inner sep=2,midway, fill=white, text=red] {tw $\Omega(n)$} (10,2.05);
\end{tikzpicture}
\caption{The satisfiability and primal treewidth thresholds for $\calF_2(n,\delta n)$.}
\label{figure:satisfiability_and_treewidth_threshold}
\end{figure}
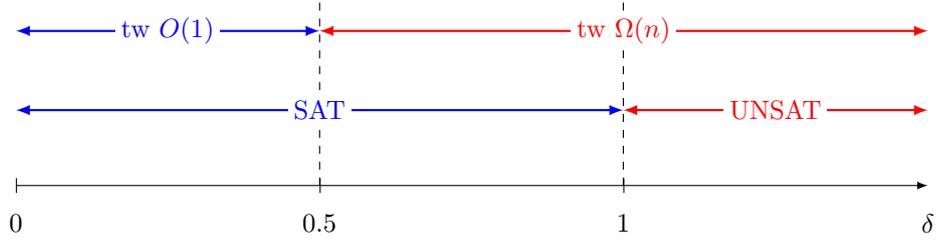

Figure~\ref{figure:satisfiability_and_treewidth_threshold} provides a visual summary of Theorem~\ref{theorem:sat_threshold_for_2CNF} and Lemma~\ref{lemma:tw_threshold_for_2CNF}.

\section{Compilability Thresholds}

The property of a 2-CNF we are interested in is its OBDD-size, that is, the size of the smallest OBDD that computes it. \medskip

\emph{Binary decision diagrams} (BDD) are well-known decision diagram representations of Boolean functions. A BDD is a directed acyclic graph with a single source and one or two sinks labeled $0$ and $1$. Each internal node is labeled with a Boolean variable and has exactly two outgoing edges called \emph{$0$-edge} and \emph{$1$-edge}. The set of variables labeling nodes of a BDD $B$ is denoted $\var(B)$. A BDD is \emph{ordered} (OBDD) when, on every source-to-sink path, each variable appears at most once and always in the same order. This total order is called the \emph{variable order} of the OBDD. In a BDD $B$, every complete assignment $\alpha$ to the variables corresponds to a single path: we start from the source and follow the $\alpha(x)$-edge for every node labeled with a variable $x$ that is reached; the label of the sink reached is the value $B(\alpha)$ computed by $B$ on $\alpha$. Note that $B(\alpha)$ is also defined for $\alpha$ an assignment to any superset $X$ of $\var(B)$. Let $f$ be a Boolean function over $X$. Then we say $B$ computes $f$ when $B(\alpha) = f(\alpha)$ for every $\alpha \in \{0,1\}^X$. 

\begin{definition}[OBDD-size]
The size of an OBDD is its number of nodes, including sinks. $\obddSize(f)$ is the minimal size of an OBDD computing the Boolean function $f$.
\end{definition}

\noindent The OBDD-size of a Boolean formula $F$, written $\obddSize(F)$, is the OBDD-size of the Boolean function over $\var(F)$ that maps to $1$ an assignment to $\var(F)$ if and only if that assignment is in $\sat(F)$. Contrary to satisfiability or primal treewidth, OBDD-size is not monotone in the sense that, for a CNF $F$ and a subformula $F'$ of $F$, we can have $\obddSize(F) \leq \obddSize(F')$ or $\obddSize(F) \geq \obddSize(F')$.

\subsection{Compilability Threshold for Monotone $k$-CNF}

As a warm-up, we consider \emph{monotone formulas}.
A monotone CNF formula is a CNF formula where every literal is positive. We denote by $\calF^m_2(n,\delta n)$ the uniform distribution for monotone 2-CNF over variables $X_n$ containing exactly $\delta n$ \emph{monotone clauses} of length $2$. Since monotone formulas are always satisfiable, their OBDD-sizes do not exhibit a threshold behavior at $\delta = 1$. Understanding the compilability threshold is arguably easier for $\calF^m_2(n,\delta n)$ than for $\calF_2(n,\delta n)$ since the OBDD-size of monotone CNF is fairly well-understood already. 

For $G$ a graph, we denote by $\Delta(G)$ the maximum degree of a vertex of $G$.

\begin{theorem}[{\cite[Theorem 7.1]{amarilli2020connecting}}]\label{theorem:OBDD_size_bound_monotone_CNF}
Let $k \geq 1$. For every $n$-variable monotone $2$-CNF $F$, 
the OBDD-size of $F$ is at least 
$2^{\pw(G_F)/ (8\Delta(G_F)^2)} / n$.
\end{theorem} 
\cite[Theorem 7.1]{amarilli2020connecting} actually gives a lower bound for non-deterministic OBDDs, which are always smaller than OBDDs. It provides a lower bound on the \emph{width} of the BDD (the maximum number of nodes labeled with the same variable) and supposes \emph{completeness} (no variable is skipped along any path). To get theorem~\ref{theorem:OBDD_size_bound_monotone_CNF}, we use that the complete-OBDD-size of a function is at most $n$ times its OBDD-size, and that size is greater than width.

When $F \sim \calF_2^m(n,\delta n)$ we have that $G_F \sim \calGr(n,\delta n)$. We have $\Delta(G_F) < \log(n)$ with high probability. Combined with Theorem~\ref{theorem:tw_threshold_for_graphs_uniform}, this allows us to show that $\delta =1/2$ is a compilability threshold for $\calF^m_2(n,\delta n)$.

\begin{restatable}[$\star$]{lemma}{smallDegree}\label{lemma:small_degree}
Let $G \sim \calGr(n,\delta n)$ with $\delta \geq 0$ a constant, then $\lim\limits_{n \rightarrow \infty}\Pr[\Delta(G) \geq \log(n)] = 0$.
\end{restatable}

\begin{theorem}[Compilability Thresholds for monotone 2-CNF]\label{theorem:compilability_threshold_monotone}
Let $\delta \geq 0$ be a constant, $n \in \mathbb{N}$ and $F \sim \calF^{m}_2(n,\delta)$. There are constants $c,d > 0$ such that, 
\begin{enumerate}[(I)]
\item if $\delta < 1/2$, then $\lim_{n \rightarrow \infty} \Pr[\obddSize(F) \leq c n] = 1$;\label{item:mI}
\item if $1/2 < \delta$, then $\lim_{n \rightarrow \infty} \Pr[\obddSize(F) \geq \exp(d\frac{n}{\log^2(n)})] = 1$.;\label{item:mII}
\end{enumerate}
\end{theorem}
\begin{proof}

Since $G_F \sim \calGr(n,\delta n)$, the case $\delta > 1/2$ follows from Theorem~\ref{theorem:tw_threshold_for_graphs_uniform},  Theorem~\ref{theorem:OBDD_size_bound_monotone_CNF},   Lemma~\ref{lemma:small_degree} and $\tw(G) \leq \pw(G)$. For $\delta < 1/2$, since $\pw(G) \le O(\tw(G)\log(n))$~\cite[Corollary 24]{Bodlaender98},
by Theorem~\ref{theorem:tw_threshold_for_graphs_uniform} there is a constant $d > 0$ such that $\lim_{n \rightarrow \infty} \Pr[\pw(G_F) \leq d \log(n)] = 1$. The existence of a constant $c > 0$ such that $\lim_{n \rightarrow \infty} \Pr[\obddSize(F) \leq c n] = 1$ then  follows from the well-known $O(n2^{\pw(G_F)})$ upper bound on the OBDD-size of $F$~\cite[Theorem~2.1]{ferrara2005treewidth}.
\end{proof}
 Theorem~\ref{theorem:compilability_threshold_monotone} (\ref{item:mII}) may be surprising at first glance. Since the expected number of solutions decreases as $\delta$ increases, one could suspect that, for $\delta$ large enough, most formulas in $\calF^m_2(n,\delta n)$ have so few solutions that they admit small-size OBDDs. But this argument requires $\delta$ to be a function of $n$ (for dense CNF), whereas  we fix $\delta$ to a constant (for sparse CNF).

\subsection{Compilability Thresholds for $2$-CNF}

Our main result is the proof that, in the case of general 2-CNF formulas drawn $\calF_2(n,\delta n)$, there are two compilability thresholds, namely $\delta = 1/2$ and $\delta = 1$. Figure~\ref{figure:compilability_threshold} provides a visual summary of this phenomenon. It is no coincidence that these are exactly the satisfiability threshold and the primal treewidth threshold.  We restate Theorem~\ref{theorem:compilability_thresholds} for the reader's convenience.  

\compilabilityThreshold*
\medskip
 
(\ref{item:III}) is immediate from the satisfiability threshold and (\ref{item:I}) is proved in almost the same way as (\ref{item:mI}) for Theorem~\ref{theorem:compilability_threshold_monotone}.  Proving (\ref{item:II}) requires a lot more work and is essentially the content of the next four sections.

\begin{proof}[Proof of Theorem~\ref{theorem:compilability_thresholds} (\ref{item:I})]
Let $\delta < 1/2$ and $F \sim \calF_2(n, \delta n )$. By Lemma~\ref{lemma:tw_threshold_for_2CNF}, $\lim_{n \rightarrow \infty} \Pr[\tw(F) \leq 2] = 1$. Since $\pw(F) \leq O(\log(n)\tw(F))$~\cite[Corollary 24]{Bodlaender98}, this means that $\lim_{n \rightarrow \infty} \Pr[\pw(F) \leq O(\log(n))] = 1$. Using that the OBDD-size of $F$ is at most $O(n2^{pw(F)})$~\cite[Theorem~2.1]{ferrara2005treewidth} we get $\lim_{n \rightarrow \infty} \Pr[\obddSize(F) \leq n^c] = 1$ for some $c > 0$.
\end{proof}

\begin{proof}[Proof of Theorem~\ref{theorem:compilability_thresholds} (\ref{item:III})] Let $\delta > 1$ and $F \sim \calF_2(n, \delta n )$. Then Theorem~\ref{theorem:sat_threshold_for_2CNF} tells us that $\lim_{n \rightarrow \infty} \Pr[F \in \SAT] = 0$. The OBDD-size of an unsatisfiable formula is $1$ (the OBDD is just the $0$-sink), hence $\lim_{n \rightarrow \infty} \Pr[\obddSize(F) = 1] = 1$. 
\end{proof}

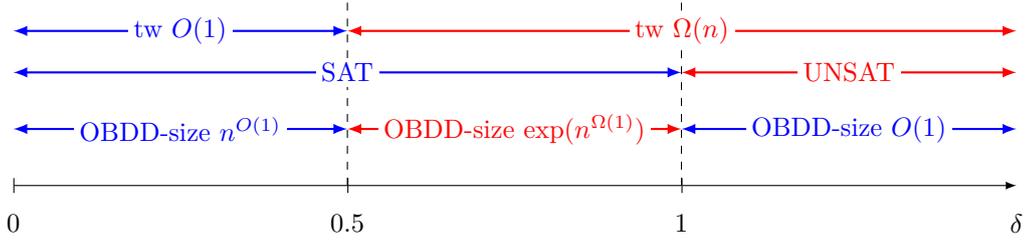
\begin{figure}
\centering
\begin{tikzpicture}[xscale=1.32]
\node (0) at (0,-0.5) {$0$};
\node (0) at (3.33,-0.5) {$0.5$};
\node (0) at (6.66,-0.5) {$1$};
\node (0) at (10,-0.5) {$\delta$};
\draw[|-] (0,0) -- (3.33,0);
\draw[|-] (3.33,0) -- (6.66,0);
\draw[|-latex] (6.66,0) -- (10,0);

\draw[dashed] (3.33,0) -- (3.33,2.5);
\draw[dashed] (6.66,0) -- (6.66,2.5);

\draw[thick,blue,latex-latex] (0,1.5) to node[inner sep=2,midway, fill=white, text=blue] {SAT} (6.66,1.5);
\draw[thick,red,latex-latex] (6.66,1.5) to node[inner sep=2,midway, fill=white, text=red] {UNSAT} (10,1.5);

\draw[thick,blue,latex-latex] (0,2.05) to node[inner sep=2,midway, fill=white, text=blue] {tw $O(1)$} (3.33,2.05);
\draw[thick,red,latex-latex] (3.33,2.05) to node[inner sep=2,midway, fill=white, text=red] {tw $\Omega(n)$} (10,2.05);

\draw[thick,blue,latex-latex] (0,0.75) to node[inner sep=2,midway, fill=white, text=blue] {OBDD-size $n^{O(1)}$} (3.33,0.75);
\draw[thick,blue,latex-latex] (6.66,0.75) to node[inner sep=2,midway, fill=white, text=blue] {OBDD-size $O(1)$} (10,0.75);
\draw[thick,red,latex-latex] (3.33,0.75) to node[inner sep=2,midway, fill=white, text=red] {OBDD-size $\exp(n^{\Omega(1)})$} (6.66,0.75);
\end{tikzpicture}
\caption{The thresholds for the OBDD-size of formulas in $\calF_2(n,\delta n)$.}
\label{figure:compilability_threshold}
\end{figure}

In the rest of the paper, we dive into the proof of Theorem~\ref{theorem:compilability_thresholds} (\ref{item:II}). Before that, we provide some intuition behind the proof. We start from an arbitrary minimal-size OBDD $B(F)$ for $F$. This OBDD uses a certain variable order $x_{\pi(1)},\dots,x_{\pi(n)}$. If we split the variable order somewhere, say at position $k$, we get a bipartition of the variables: $x_{\pi(1)},\dots,x_{\pi(k)}$ and $x_{\pi(k+1)},\dots,x_{\pi(n)}$. Variables are vertices of $G_F$, so we have a bipartition of $G_F$'s vertices. When $F$ has high primal treewidth (which happens almost always), it turns out that there is always a position $k$ such that cutting the variable order at $k$ yields a bipartition of $G_F$ such that there is a large matching between the two parts, i.e., many edges of $G_F$ cross from $x_{\pi(1)},\dots,x_{\pi(k)}$ to $x_{\pi(k+1)},\dots,x_{\pi(n)}$ without sharing a vertex (Section~\ref{section:matching}). Each edge in $G_F$ corresponds to at least one clause of $F$, so this large matching corresponds to a subformula $M(F)$ of $F$ that has many solutions. We can show that if exponentially many solutions to $M(F)$ have an extension that satisfies $F$, then the $k$th layer of $B(F)$ has exponentially many nodes (Section~\ref{section:bound}). $B(F)$ is a minimal-size OBDD for $F$, so all OBDDs for $F$ have exponentially many nodes.

To give a simple example, suppose $n = 8$, $\pi$ is the identity, we cut $\pi$ at $3$ and find a matching $\{\{x_1,x_4\},\{x_2,x_5\}, \{x_3,x_6\}\}$ in $G_F$. Say this corresponds to the subformula $M(F) = (x_1 \lor x_4) \land (x_2 \lor x_5) \land (x_3 \lor x_6)$ in $F$ and say all $3^3$ solutions to $M(F)$ have an extension satisfying $F$. We claim that the two assignments $\alpha = (1,0,0)$ and $\beta = (0,0,0)$ to $(x_1,x_2,x_3)$ do not reach the same node in $B(F)$, because
\begin{itemize}
\item $F|\beta$ forces $x_4$ to $1$ due to $x_1 \lor x_4$ while 
\item $F|\alpha$ does not force $x_4$ to $1$ or $0$ since both assignments $(1,0,0,1,1,1)$ and $(1,0,0,0,1,1)$ to $(x_1,x_2,x_3,x_4,x_5,x_6)$ have extensions that satisfy $F$.
\end{itemize}
In fact, in this case, any two assignments to $(x_1,x_2,x_3)$ reach a different node, so there are $2^3$ different nodes in $B(F)$. In Section~\ref{section:bound} we make this argument work when only a constant fraction of $M(F)$'s solutions have an extension satisfying $F$.

Then comes the issue of showing this nice property of $M(F)$. For that, we rely on a recent theorem of~\cite{basse2025regularity,rasmussen2025fixing} that says that for $\delta < 1$, with high probability on $F$, a random assignment of length $o(\sqrt{n})$ is very likely to have an extension satisfying $F$. The problem is that this random assignment is independent of $F$ while solutions to $M(F)$ are not. To overcome this issue, we move the analysis to random formulas sampled from another distribution in which every clause is independent of the others, at the risk of creating duplicate clauses (Section~\ref{section:distribution}).

\section{Random 2-CNF with Clause Replacement}\label{section:distribution}

\subsection{The Distribution $\calH_2(n,m)$}

We are going to prove Theorem~\ref{theorem:compilability_thresholds} (\ref{item:II}) through another $2$-CNF distribution. Given $n, m \in \mathbb{N}$, $\calH_2(n,m)$ is the distribution for the random CNF formula $F = C_1 \land \dots \land C_m$ where each clause $C_i$ is drawn uniformly at random and independently from the set of all $4\binom{n}{2}$ clauses of length $2$ over $X_n = \{x_1,\dots,x_n\}$. For $k \in [m]$ let $F^{\leqslant k} = \bigwedge_{i = 1}^k C_i$ and $F^{> k} = \bigwedge_{i = k+1}^m C_i$. 
$$
F = F^{\leqslant k} \land F^{> k}
$$
A neat property of $\calH_2(n,m)$ is that $F^{\leqslant k} \sim \calH_2(n,k)$ and $F^{> k} \sim \calH_2(n,m-k)$. Even nicer, $F^{\leqslant k}$ is stochastically independent of $F^{> k}$. The main drawback of $\calH_2(n,m)$, on the other hand, is that the same clause may be drawn more than once. We call a CNF formula $F$ \emph{simple} when no two of its clauses are identical. It is not unlikely that $F$ has duplicate clauses, but for $m = \delta n$ with $\delta > 0$ fixed, the probability that $F$ is simple does not vanish to $0$ as $n$ increases. Let $N = 4\binom{n}{2}$. This is the well-known ``birthday problem'': there are $N$ days in a year ($N$ possible clauses), and we draw $m$ people and thus $m$ birthdays ($m$ clauses) at random, what is the probability that two people have the same birthday (that two clauses are identical)? 
Here, since $m=o(N)$ as $n \to\infty$, the probability that $F$ is simple (no two people share a birthday) converges to a positive constant, and in particular remains bounded away from $0$. We can reuse the analysis of~\cite[Section 1.3]{frieze2015introduction}:
\begin{equation}\label{eq:proba-simple}
\begin{aligned}
\Pr\limits_{F \sim \calH_2(n,m)}[F \text{ is simple}] &\geq \frac{\binom{N}{m}m!}{N^m} \geq (1 - o(1))\frac{N^m}{m!}\exp\left(-\frac{m^2}{2N} - \frac{m^3}{6N^2}\right)\frac{m!}{N^m}
\\
&\geq (1 - o(1))\exp\left(-\frac{m^2}{2N} - \frac{m^3}{6N^2}\right) \geq (1-o(1))e^{-\delta^2 - \delta}.
\end{aligned}
\end{equation}
Another crucial feature of $\calH_2(n,m)$ is that a random formula $F$ drawn from $\calH_2(n,m)$, when conditioned on being simple, behaves as if it were drawn from $\calF_2(n,m)$. Formally, let $\mathscr{P}$ be a formula property (a set of formulas), then 
\begin{equation}\label{eq:proba-conditioned}
\Pr\limits_{F \sim \calF_2(n,m)}[F \in \mathscr{P}] = \Pr\limits_{F \sim \calH_2(n,m)}[F \in \mathscr{P} \mid F \text{ is simple}]
\end{equation}
The combination of (\ref{eq:proba-simple}) and (\ref{eq:proba-conditioned}) in a Bayes rule shows that, if the event $F \in \mathscr{P}$ is asymptotically negligible when $F$ is drawn from $\calH_2(n,m)$, then the same event is asymptotically negligible when $F$ is drawn from $\calF_2(n,m)$ (provided $m = \delta n$ and $\delta > 0$). 
\begin{restatable}[$\star$]{lemma}{fromHtoF}\label{lemma:from_H_to_F}
Let $\delta > 0$.  If $\lim\limits_{n \rightarrow \infty}\Pr\limits_{\substack{F \sim \\ \calH_2(n,\delta n)}}[F \in \mathscr{P}] = 0$ then \mbox{$\lim\limits_{n \rightarrow \infty}\Pr\limits_{\substack{F  \sim \\ \calF_2(n,\delta n)}}[F \in \mathscr{P}] = 0$}.
\end{restatable}

\noindent Thus, to prove Theorem~\ref{theorem:compilability_thresholds} (\ref{item:II}), it is enough to show the correctness of the following lemma.
\begin{lemma}\label{lemma:target_lemma}
Let $1/2 < \delta < 1$ be constant and $F \sim \calH_2(n,\delta n)$. There exists $v > 0$ such that 
$$
\lim\limits_{n \rightarrow \infty}\Pr\left[\text{OBDD-size}(F) \leq 2^{|\var(F)|^v}\right] = 0.
$$
\end{lemma}

\begin{proof}[Proof of Theorem~\ref{theorem:compilability_thresholds} (\ref{item:II})]
Combining Lemmas~\ref{lemma:from_H_to_F} and~\ref{lemma:target_lemma} with $\scrP$ the class of CNF formulas such that $F \in \scrP$ if and only if $\text{OBDD-size}(F) \leq 2^{|\var(F)|^v}$, we get $\lim_{n \to \infty} \Pr_{F \sim \calF_2(n,\delta n)}[F \in \scrP] = 0$, 
i.e., $\lim_{n \to \infty} \Pr_{F \sim \calF_2(n,\delta n)}
[\text{OBDD-size}(F) \leq 2^{|\var(F)|^v}] = 0$.
\end{proof}

The objective now is to prove Lemma~\ref{lemma:target_lemma}.

\subsection{High Treewidth in $\calH_2(n,\delta n)$}

Lemma~\ref{lemma:tw_threshold_for_2CNF} tells us that a random CNF $F$ almost always has high treewidth when drawn from $\calF_2(n,\delta n)$ with $\delta > 1/2$, but is it also true when $F$ is drawn from $\calH_2(n,\delta n)$? This question has to be addressed because high treewidth is an essential component in the proof of Lemma~\ref{lemma:target_lemma}. To prove that high treewidth indeed occurs almost always in $F \sim \calH_2(n,\delta n)$, we show that $F$ is unlikely to contain more than $\sqrt{n}$ duplicate clauses. This will mean that the primal graph of $F \sim \calH_2(n,\delta n)$ is essentially a graph $G \sim \calGr(n,\delta n)$ from which we remove at most $\sqrt{n}$ vertices. But then, since $G$ has treewidth $\Omega(n)$ with high probability (Theorem~\ref{theorem:tw_threshold_for_graphs_uniform}), we will have that the graph of $F$ has treewidth at least $\Omega(n) - \sqrt{n} = \Omega(n)$ with high probability. 

So let us show that we do not have too many duplicate clauses in $F = C_1 \land \cdots \land C_m \sim \calH_2(n,\delta n)$. For $k \in [m]$ fixed, the probability that the clause $C_k$ is identical to a clause in $C_1,\dots,C_{k-1},C_{k+1},\dots,C_m$ is $1 - \left(\frac{N-1}{N}\right)^{m-1}$, with $N= 4 \binom{n}{2} = 4n^2-4n$, so the expected number of non-unique clauses (the number of people with non-unique birthday) is 
\begin{equation}\label{eq:birthday}
m\left(1 - \left(\frac{N-1}{N}\right)^{m-1}\right) = \delta n \left(1 - \left(\frac{2n^2-2n-1}{2n^2-2n}\right)^{\delta n-1}\right)  \leq \delta n \left(1 - \left(1-\frac{1}{n^2}\right)^{\delta n}\right).
\end{equation}
Using the Markov bound, we then show that having more than $\sqrt{n}$ non-unique clauses is getting very unlikely as $n$ increases.
\begin{restatable}[$\star$]{lemma}{notTooManyNonUniqueClauses}\label{lemma:not-too-many-non-unique-clauses-in-H}
Let $\delta > 0$ and $F \sim \calH_2(n,\delta n)$. Then
$$\lim\limits_{n \rightarrow \infty} \Pr[F \text{ has at least } \sqrt{n} \text{ non-unique clauses}]  = 0.
$$
\end{restatable}

\noindent We can now prove that, when $\delta > 1/2$, the treewidth of $F$ is in $\Omega(n)$ with high probability using the argument explained above.

\begin{restatable}[$\star$]{lemma}{largeTWinH}\label{lemma:large-tw-in-H}
Let $\delta > 1/2$ and $F \sim \calH_2(n,\delta n)$. There is a constant $\gamma > 0$ such that 
$$
\lim\limits_{n \rightarrow \infty} \Pr[\tw(G_F) \geq \gamma n]  = 1.
$$
\end{restatable}

\section{From Large Treewidth to Large Matchings}\label{section:matching}

We recall that a matching in a graph $G$ is a subset $M \subseteq E(G)$ of its edges such that no two edges of $M$ share an endpoint. Let $V(M)$ be the set of endpoints in $M$. 
For $V_1,V_2$ two disjoint subsets of $V(G)$, a matching in $G$ between $V_1$ and $V_2$ is a matching of $G$ whose edges all have one endpoint in $V_1$ and the other in $V_2$. We denote by $mm_G(V_1,V_2)$ the maximum size of a matching between $V_1$ and $V_2$ in $G$.

A key ingredient for the proof of Theorem~\ref{theorem:compilability_thresholds} (\ref{item:II}) is the presence of a large matching in the primal graph $G_F$ that corresponds to a split in the variable order of $B(F)$. Suppose every path in $B(F)$ reads all variables and that we count the nodes at level $k$ of the OBDD. For a clause $a \lor b \in F$ with $a$ read before level $k$ and $b$ read after level $k$, any two satisfying assignments $\alpha$ and $\beta$ such that $\alpha(a) = 1$, $\alpha(b) = 0$ and $\beta(a) = 0$, $\beta(b) = 1$ have to reach different nodes at level $k$, for otherwise $B(F)$ would accept an assignment that falsifies $a \lor b$. To derive that there are many nodes at level $k$ using this idea, we need many disjoint clauses with one variable read on each side of level $k$; thus, we want a lage matching in $G_F$ between the set $V_1$ of variables read before level $k$ and the set $V_2$ of variable read after level $k$. In fact, any $k$ works as long as the corresponding partition in $G_F$ gives a large matching. To find this large matching, we use the connection between treewidth and maximum matching width.

\subsection{Large Matchings in $\calH_2(n,\delta n)$}

We follow the definitions from~\cite{Vatshelle2012}. A \emph{binary decomposition tree} $T$ of $G$ is a rooted binary tree $T$ whose leaves are in bijection with $V(G)$. For $t \in V(T)$ let $V_t \subseteq V(G)$ be the vertices of $G$ corresponding to the leaves below $t$ in $T$ and let $\bar V_t$ be the vertices of $G$ corresponding to the remaining leaves, i.e., $\bar V_t = V(G) \setminus V_t$. $(V_t,\bar V_t)$ is a bipartition of $V(G)$. 
\begin{definition}[Maximum Matching Width]
Let $G$ be a graph and $T$ be a binary decomposition tree of $G$. 
\begin{itemize}
\item $\displaystyle \mmw(G,T) := \max_{t \in V(T)} mm_G(V_t,\bar V_t)$.
\item $\displaystyle \mmw(G) := \min_T \mmw(G,T)$ with $T$ ranging over all binary decomposition trees of $G$.
\end{itemize}
We call $\mmw(G)$ the \emph{maximum  matching width} of $G$.
\end{definition}

\begin{theorem}[\cite{Vatshelle2012}, Theorem 4.2.5]\label{theorem:tw_mmw}
Let $G$ be a graph, then $\frac{1}{3}\tw(G) \leq \mmw(G) \leq \tw(G)+1$.
\end{theorem}

\noindent Combining Lemma~\ref{lemma:large-tw-in-H} and Theorem~\ref{theorem:tw_mmw} then yields the following.

\begin{restatable}[$\star$]{lemma}{largeMimwInH}\label{lemma:large-mimw-in-H}
Let $\delta > 1/2$ and $F \sim \calH_2(n,\delta n)$. There is a constant $\gamma > 0$ such that 
$$
\lim\limits_{n \rightarrow \infty} \Pr[\mmw(G_F) \geq \gamma n]  = 1.
$$
\end{restatable}

For $F \sim \calH_2(n,\delta n)$, let $B(F)$ be a minimal-size OBDD representing $F$, breaking ties arbitrarily. Let $\pi(F) : (x_{\pi(1)},x_{\pi(2)},\dots,x_{\pi(n)})$ be the variable order used by this OBDD. Variable orders are in one-to-one correspondence with binary tree decompositions in which the tree is right-linear, i.e., every internal node has a leaf as its left child. Let $T(F)$ be the tree decomposition for $\pi(F)$. The picture below provides an example for $n = 5$.
\begin{center}
\begin{tikzpicture}[xscale=2,yscale=2]
\node[draw,circle,inner sep=2pt] (a) at (0,0) {};
\node[draw,circle,inner sep=2pt] (b) at (0.5,-0.25) {};
\node[draw,circle,inner sep=2pt,label={north east:$t$}] (c) at (1,-0.5) {};
\node[draw,circle,inner sep=2pt] (d) at (1.5,-0.75) {};

\node[font=\footnotesize,inner sep=2pt] (x1) at (-0.5,-0.25)  {$x_{\pi(1)}$};
\node[font=\footnotesize,inner sep=2pt] (x2) at (0,-0.5) {$x_{\pi(2)}$};
\node[font=\footnotesize,inner sep=2pt] (x3) at (0.5,-0.75) {$x_{\pi(3)}$};
\node[font=\footnotesize,inner sep=2pt] (x4) at (1,-1) {$x_{\pi(4)}$};
\node[font=\footnotesize,inner sep=2pt] (x5) at (2,-1) {$x_{\pi(5)}$};

\draw (a) -- (b) -- (c) -- (d);
\draw (a) -- (x1);
\draw (b) -- (x2);
\draw (c) -- (x3);
\draw (d) -- (x4);
\draw (d) -- (x5);
\end{tikzpicture}
\end{center}
If $t$ is a leaf of $T(F)$ then $|V_t| = 1$ and the maximum matching between $V_t$ and $\bar V_t$ has size $0$ or $1$. So, generally, $mmw(G_F,T(F))$ is obtained for some internal node of $T(F)$. Selecting an internal node $t$ in $T(F)$ gives a bipartition $(\bar V_t, V_t)$ that amounts to cutting $\pi(F)$. For instance, in the above example, $(\bar V_t, V_t) = (\{x_{\pi(1)},x_{\pi(2)}\},\{x_{\pi(3)},x_{\pi(4)},x_{\pi(5)}\})$. We have a partition of the variables $\Pi(F) = (\Pi_1(F),\Pi_2(F))$ obtained by splitting $\pi(F)$ such that there is a matching $M$ of size $|M| = \mmw(G_F,T(F))$ between $\Pi_1(F)$ and $\Pi_2(F)$ in $G_F$ (breaking ties arbitrarily for $\Pi(F)$). Each edge in $M$ corresponds to at least one clause in $F$, so we let $M(F)$ be a subformula of $F$ obtained by taking one such clause per edge in $M$ (again, breaking ties arbitrarily). Note that $|M(F)| = |M|$ and $G_{M(F)} = M$. We shall keep in mind that the variables $B(F)$, $\pi(F)$, $\Pi(F)$, $T(F)$, $M(F)$ are all defined \emph{deterministically} from $F$. We technically also need that these variables are equal for any two formulas $F$ and $F'$ that differ only up to permutation of the clauses. The next corollary follows from Lemma~\ref{lemma:large-mimw-in-H}.
 
\begin{restatable}[$\star$]{corollary}{largeMatchingFormula}\label{corollary:large-matching}
Let $\delta > 1/2$ and $F \sim \calH_2(n,\delta n)$. There is a constant $\gamma > 0$ such that 
$$
\lim\limits_{n \rightarrow \infty} \Pr\left[|M(F)| \geq \gamma n\right] = 1.
$$
\end{restatable}

\subsection{Matching on the First $n^{1/3}$ Clauses}

Now we have $M(F)$, our large matching subformula of $F$. Here we explain that, since $|M(F)|$ is in $\Omega(n)$ with high probability, that is, $M(F)$ makes a constant fraction of all $\delta n$ clauses, we can expect a good fraction of the \emph{first clauses} to be in $M(F)$ (recall that the clauses of $F \sim \calH_2(n,\delta n)$ are ordered). We consider the first $n^{1/3}$ clauses.
\begin{definition}[Matching formula]
Given a variable bipartition $\Pi = (\Pi_1,\Pi_2)$, a \emph{$\Pi$-matching formula} is a $2$-CNF $\bigwedge_{i = 1}^h (\ell_i \lor \ell_{h+i})$ such that $\var(\ell_i) \neq \var(\ell_j)$ for every $i \neq j$ in $[2h]$ and such that $\var(\ell_i) \in \Pi_1$ and $\var(\ell_{h+i}) \in \Pi_2$ for all $i \in [h]$.
\end{definition}
Alternatively, a matching formula is a 2-CNF with $m$ clauses and a primal graph that is composed of $m$ pairwise disjoint edges. For instance, $(x_1 \lor \neg x_3) \land (x_2 \lor x_5) \land (\neg x_4 \lor \neg x_6)$ is a $(\{x_1,x_2,x_4\},\{x_3,x_5,x_6\})$-matching formula with primal graph 

\begin{center}\begin{tikzpicture}[scale=1, rotate=90]
\node (a) at (0,0) {$x_1$};
\node (b) at (1,0) {$x_3$};
\node (c) at (0,-1) {$x_2$};
\node (d) at (1,-1) {$x_5$};
\node (e) at (0,-2) {$x_4$};
\node (f) at (1,-2) {$x_6$};
\draw (a) -- (b);
\draw (c) -- (d);
\draw (e) -- (f);
\end{tikzpicture}
\end{center}
On the other hand, $(x_1 \lor x_3) \land (x_1 \lor \neg x_3) \land (x_2 \lor x_5) \land (\neg x_4 \lor \neg x_6)$  is not a matching formula, although it has the same primal graph. 

Let $n \geq 2k$.  We denote by $\MF_{n,k}$ the set of matching formulas with exactly $k$ clauses that use only variables in $x_1,\dots,x_n$. For instance, $(x_1 \lor \neg x_3) \land (x_2 \lor x_5) \land (\neg x_4 \lor \neg x_6)$ belongs to $\MF_{6,3}$, but also to $\MF_{7,3}$, $\MF_{8,3}$, etc. Counting the solutions to a matching formula is straightforward: if $F \in \MF_{n,k}$, then $|\var(F)| = 2k$ and $|\sat(F)| = 3^k$.

We now split $F \sim \calH_2(n,\delta n)$ into $F^{\leq k}$ and $F^{> k}$ with
$
k = n^{1/3}
$
and consider the probability that $F^{\leq k}$ is a matching formula and that $M(F)$ intersects $F^{\leq k}$ on a large section. 

\begin{restatable}[$\star$]{lemma}{firstKclausesAreMatching}\label{lemma:the_first_k_clauses_in_H_are_a_matching}
Let $\delta > 0$, $F \sim \calH_2(n,\delta n)$ and $k = o(\sqrt{n})$. Then 
$$
\lim\limits_{n \rightarrow \infty} \Pr\left[F^{\leq k} \in \MF_{n,k}\right] = 1.
$$
\end{restatable}

Let $M(F,k)$ be the intersection of $M(F)$ with $F^{\leq k}$. For a fixed $F$, we have $M(F) = M(F')$ for every $F'$ obtained by randomly permuting the clauses of $F$. We claim that, when $|M(F)| \geq \Omega(n)$, a constant fraction of $(F')^{\leq k}$ is likely to intersect $M(F)$. Since $F'$ and $F$ are equally likely to be drawn from $\calH_2(n,\delta n)$, we can then show the following lemma.

\begin{restatable}[$\star$]{lemma}{matchingOnFirstKclauses}\label{lemma:the_first_k_clauses_in_H_contain_a_big_PiF-matching}
Let $\delta > 1/2$, $F \sim \calH_2(n,\delta n)$ and $k = n^{1/3}$  There is a constant $\gamma > 0$ such that 
$$
\lim\limits_{n \rightarrow \infty} \Pr\left[|M(F,k)| \geq  \gamma k\right] = 1.
$$
\end{restatable}

\section{Lower Bounds on the OBDD-Size}\label{section:bound}

In this section, $F$ is a fixed 2-CNF formula. Recall that we have a fixed OBDD $B(F)$ whose size we want to bound from below and that the variable bipartition $\Pi(F) = (\Pi_1(F),\Pi_2(F))$ is obtained by splitting the variable order of $B(F)$.

\begin{lemma}\label{lemma:B_has_size_2_to_A}
Let $F$ be a 2-CNF and $H = \bigwedge_{i = 1}^{h} (\ell_i \lor \ell_{h+i})
$ be a $\Pi(F)$-matching subformula of  $F$. Let $\calA$ be a subset of $\sat(H)$ such that for all $\alpha \in \calA$, we have $F \land \alpha \in \SAT$ and such that, for all $\alpha \neq \beta$ in $\calA$, there is $i \in [h]$ such that $\alpha(\ell_i) \neq \beta(\ell_i)$ and $\alpha(\ell_{h+i}) \neq \beta(\ell_{h+i})$. Then $B(F)$ has at least $|\calA|$ nodes.
\end{lemma}
\begin{proof}
For every $\alpha \in \calA$, there is an extension of $\alpha$ that satisfies $F$; we call the extended assignment $\alpha^*$. We denote by $\alpha^*_1$ the restriction of $\alpha^*$ to $\Pi_1(F)$ and by $\alpha^*_2$ the restriction of $\alpha^*$ to $\Pi_2(F)$. The variables $\var(\ell_1),\dots,\var(\ell_h)$ are assigned by $\alpha^*_1$ whereas the variables $\var(\ell_{h+1}),\dots,\var(\ell_{2h})$ are assigned by $\alpha^*_2$. 

Consider $\alpha,\beta \in \calA$ with $\alpha \neq \beta$. The assignments $\alpha^*$ and $\beta^*$ are in $\sat(F)$, so their corresponding paths in $B(F)$ reach the $1$-sink. Suppose the paths for $\alpha^*_1$ and $\beta^*_1$ reach the same node in $B(F)$, then the paths for $\alpha^*_1 \cup \beta^*_2$ and $\beta^*_1 \cup \alpha^*_2$ also reach the $1$-sink in $B(F)$ and therefore $\alpha^*_1 \cup \beta^*_2 \in \sat(F)$ and $\beta^*_1 \cup \alpha^*_2 \in \sat(F)$. But there exists $i \in [h]$ with $\alpha(\ell_i) \neq \beta(\ell_i)$ and $\alpha(\ell_{h+i}) \neq \beta(\ell_{h+i})$, which implies that $\alpha(\ell_{i}) = \beta(\ell_{h+i}) = 0$ or $\beta(\ell_{i}) = \alpha(\ell_{h+i}) = 0$, which in turn implies that $\alpha^*_1(\ell_i) = \beta^*_2(\ell_{h+i}) = 0$ or $\beta^*_1(\ell_i) = \alpha^*_2(\ell_{h+i}) = 0$. Thus, either $\alpha^*_1 \cup \beta^*_2$ or $\beta^*_1 \cup \alpha^*_2$ falsifies $\ell_i \lor \ell_{i+h}$ and therefore falsifies $F$. This is a contradiction, so for every $\alpha,\beta \in \calA$ with $\alpha \neq \beta$, we find that $\alpha^*_1$ and $\beta^*_1$ do not reach the same node in $B(F)$. It follows that $B(F)$ has at least $|\calA|$ nodes.
\end{proof}

For two formulas $H$ and $F$, we let $\theta(H,F)$ be the fraction of assignments $\alpha\in \sat(H)$ that have an extension satisfying $F$. Formally,
$$
\theta(H,F) = \frac{|\{\alpha \in \sat(H) : \alpha \land F \in \SAT\}|}{|\sat(H)|}.
$$
If $H$ is unsatisfiable, then we define $\theta(H,F) = 1$. Note that $\alpha$ is an assignment over the variables of $H$. In particular, $\theta(H,F) \geq 2/3$ means that at least $2/3^\text{rd}$ of the satisfying assignments of $H$ have an extension that satisfies $F$. If $H$ is satisfiable and $F$ is unsatisfiable, then $\theta(H,F) = 0$. If both $H$ and $F$ are unsatisfiable then, by convention, $\theta(H,F) = 1$. 

\begin{lemma}\label{lemma:big_A_if_many_model_for_F_M}
Let $F$ be a 2-CNF and $H = \bigwedge_{i = 1}^{h} (\ell_i \lor \ell_{h+i})$ be a $\Pi(F)$-matching subformula of $F$. Consider $\calA' \subseteq \sat(H)$ such that, for all $\alpha \in \calA'$, $F \land \alpha \in \SAT$ and $\alpha(\ell_i) = \alpha(\ell_{h+i}) = 1$ holds for \underline{no more than} $h/2$ indices $i$. Then there is $\calA \subseteq \calA'$ that verifies the assumptions of Lemma~\ref{lemma:B_has_size_2_to_A} and
$$
|\calA| \geq \frac{|\calA'|}{3^{h/2}2^{h/2}}.
$$
\end{lemma}
\begin{proof}
All assignments $\alpha$ in $\calA'$ satisfy $H$, so for each $i \in [h]$ we have three possibilities for $(\alpha(\ell_i),\alpha(\ell_{h+i}))$, namely $(1,1)$, $(0,1)$ and $(1,0)$. We denote by $I_\alpha$ the set of $i$ such that $\alpha(\ell_i) \neq \alpha(\ell_{h+i})$. By assumption, we have 
$
|I_\alpha| \geq h/2$ for every $\alpha \in \calA'$.
We start with $\calA = \emptyset$ and fill it by repeating the following two steps until $\calA'$ is empty: 
\smallskip
\begin{itemize}
\item select $\alpha$ in $\calA'$ and add it to $\calA$;
\item remove from $\calA'$ all assignments $\beta$ such that, for every $i \in I_\alpha$ we have $(\beta(\ell_i),\beta(\ell_{h+i})) = \alpha(\ell_i),\alpha(\ell_{h+i}))$ or $(\beta(\ell_i),\beta(\ell_{h+i})) = (1,1)$.
\end{itemize}  
\smallskip
The assignments removed from $\calA'$ in the second step are said to be \emph{removed because of $\alpha$}. Note that the second step removes at least $\alpha$ from $\calA'$ so the process terminates. Thanks to the removal, after $\alpha$ is added to $\calA$ we are guaranteed that all future assignments $\gamma$ added to $\calA$ verify $\gamma(\ell_i) \neq \alpha(\ell_i)$ and $\gamma(\ell_{h+i}) \neq \alpha(\ell_{h+i})$ for some $i \in I_\alpha$. The resulting set $\calA$ thus verifies the assumptions of Lemma~\ref{lemma:B_has_size_2_to_A}.

Now, how many assignments $\beta$ can be removed because of $\alpha$? Well, for each $i \not\in I_\alpha$ the pair $(\beta(\ell_{i}),\beta(\ell_{h+i}))$ can take all three values $(1,1)$, $(0,1)$ and $(1,0)$ but, for each $i \in I_\alpha$ the pair $(\beta(\ell_{i}),\beta(\ell_{h+i}))$ can only take two values. So 
$$
\text{at most } 3^{h - |I_\alpha|}2^{|I_\alpha|} \text{  assignments are removed because of } \alpha.
$$ 
Since $|I_\alpha| \geq h/2$ we have $3^{h - |I_\alpha|}2^{|I_\alpha|} \leq 3^{h/2}2^{h/2}$ and therefore the construction of $\calA$ goes through at least $\frac{S}{3^{h/2}2^{h/2}}$ many rounds before $\calA'$ is emptied, with $S$ the size of $\calA'$ at the beginning of the procedure. 
\end{proof}

\begin{lemma}\label{lemma:intermediate_lemma_3}
Let $F$ be a 2-CNF and $H = \bigwedge_{i = 1}^{h} (\ell_i \lor \ell_{h+i})$ be a $\Pi(F)$-matching subformula of $F$.
There is a constant $c > 0$ such that, if $\theta(H,F) \geq 2/3$, then $B(F)$ has at least $2^{c h}$ nodes.
\end{lemma}
\begin{proof}
Recall that $|\sat(H)| = 3^h$. For every $\alpha \in \sat(H)$, we say that 
\begin{itemize}
\item $\alpha$ has the property $(p_1)$ when $F \land \alpha \in \SAT$; 
\item $\alpha$ has the property $(p_2)$ when $\alpha(\ell_i) = \alpha(\ell_{h+i}) = 1$ for no more than $h/2$ indexes $i \in [h]$.
\end{itemize}
The number of assignments in $\sat(H)$ that have property $(p_2)$ is 
$$
\sum_{j = 0}^{h/2} \binom{h}{j}2^{h-j} \geq \frac{1}{2}\left(\sum_{j = 0}^{h/2} \binom{h}{j}2^{h-j} + \sum_{j = h/2+1}^{h} \binom{h}{j}2^{h-j}\right) = \frac{3^h}{2}
$$
$\theta(H,F) \geq 2/3$ ensures that the number of assignments of $\sat(H)$ that have property $(p_1)$ is at least $\frac{2}{3} 3^h$. Therefore, $\theta(H,F) \geq 2/3$ implies that there are at least $\frac{2}{3} 3^h - \frac{1}{2}3^h = \frac{1}{6} 3^h$ assignments in $\sat(H)$ with both property $(p_1)$ and property $(p_2)$.  Thus, we use Lemma~\ref{lemma:big_A_if_many_model_for_F_M} to find a set $\calA \subseteq \sat(H)$ verifying the conditions of Lemma~\ref{lemma:B_has_size_2_to_A} with 
$$
|\calA| \geq \frac{1}{6}\cdot \frac{3^h}{3^{h/2}2^{h/2}} = \frac{1}{6}\cdot\frac{3^{h/2}}{2^{h/2}} = \frac{1}{6}\cdot 2^{h\frac{\log(3)-1}{2}}.
$$
The proof is finished using Lemma~\ref{lemma:B_has_size_2_to_A}.
\end{proof}

Now, we move on with the proof of Lemma~\ref{lemma:target_lemma}. For that, we will need two more lemmas stated afterwards, but the final bound on the OBDD size follows from Lemma~\ref{lemma:intermediate_lemma_3}.

\begin{proof}[Proof of of Lemma~\ref{lemma:target_lemma}]
Let $F \sim \calH_2(n,m)$. The clauses in $F$ are ordered by construction. We consider the first $k = n^{1/3}$ clauses on one side and the remaining $m - k$ on the other. 
$$
F = F^{\leq k} \land F^{> k} 
$$
Since $k = o(\sqrt{n})$, Lemma~\ref{lemma:the_first_k_clauses_in_H_are_a_matching} ensures that $F^{\leq k}$ is a matching formula with high probability and Lemma~\ref{lemma:the_first_k_clauses_in_H_contain_a_big_PiF-matching} ensures that $F^{\leq k}$ contains a subformula $H$ that is a $\Pi(F)$-matching formula of size $\Omega(k)$ with high probability. It remains to show that $\theta(H,F) \geq 2/3$ with high probability.  Once that is proved, Lemma~\ref{lemma:target_lemma} follows by Lemma~\ref{lemma:intermediate_lemma_3}: if $\theta(H,F) \geq 2/3$ then $B(F)$ has at least 
$$
2^{\Omega(k)} \geq  2^{\Omega(n^{1/3})} 
$$ 
nodes. Since $B(F)$ is by definition a minimum-size OBDD for $F$, the statement is proved. 

The last bit, namely that $\theta(H,F) \geq 2/3$ holds with high probability, follows from Lemmas~\ref{lemma:technical-lemma} and~\ref{lemma:theta_for_subformula}. Lemma~\ref{lemma:theta_for_subformula} says that $\theta(F^{\leq k},F^{> k}) \geq 2/3$ implies $\theta(H,F) \geq 2/3$ when $F^{\leq k}$ is a matching formula, and Lemma~\ref{lemma:technical-lemma} is a ``technical lemma'' saying that $\theta(F^{\leq k},F^{> k}) \geq 2/3$ is almost always true; its proof is provided in the last technical section.
\end{proof}

\begin{restatable}{lemma}{technicalLemma}\label{lemma:technical-lemma}
Let $\tfrac{1}{2} \leq \delta < 1$, $F \sim \calH_2(n,\delta n)$ and $k = n^{1/3}$. Then 
$$
\lim\limits_{n \rightarrow \infty} \Pr\left[\theta\left(F^{\leq k},F^{> k}\right) < \tfrac{2}{3}\right] = 0.
$$
\end{restatable}

\begin{lemma}\label{lemma:theta_for_subformula}
Let $F$ be a 2-CNF, possibly with duplicate clauses. Partition the clauses of $F$ into two subformulas $F_1$ and $F_2$. Suppose $F_1$ is a matching formula and suppose $\theta(F_1,F_2) \geq 2/3$, then for every $H \subseteq F_1$ we have $\theta(H,F) \geq 2/3$.
\end{lemma}
\begin{proof}
Let $L$ be the clauses in $F_1$ that are not in $H$. Since $F_1$ is a matching formula, it has no duplicate clauses, and $H$ and $L$ are also matching formulas.  Let $f$, $h$  and $l$ be the number of clauses of $F_1$, $H$ and $L$, respectively. We have $f = h + l$ and $|\sat(F_1)| = 3^{f}$, $|\sat(H)| = 3^{h}$ and $|\sat(L)| = 3^{l}$. For $\alpha \in \sat(F_1)$, let $\alpha_H$ be its restriction to $\var(H)$ and $\alpha_L$ be its restriction to $\var(L)$. $\sat(F_1)$ is the Cartesian product of $\sat(H) \times \sat(L)$ in the sense that $\alpha \in \sat(F_1)$ if and only if the $\alpha_H \in \sat(H)$ and $\alpha_L \in \sat(L)$.  

$\theta(F_1,F_2) \geq 2/3$ implies $\theta(F_1,F) \geq 2/3$ because, if $\alpha \in \sat(F_1)$ is such that $F_2 \land \alpha \in \SAT$, then $F \land \alpha = F_1 \land F_2 \land \alpha$ is clearly in $\SAT$. Now, suppose $\theta(H,F) < 2/3$ and let $S \subseteq \sat(H)$ be the set of all $\beta \in \sat(H)$ such that $F \land \beta \in \SAT$. An assignment $\alpha \in \sat(F_1)$ can be such that $F \land \alpha \in \SAT$ only if the $F \land \alpha_H \in \SAT$. Thus, such $\alpha$'s are all contained in $S \times \sat(L)$. But then $\theta(F_1,F) \leq \frac{|S|\cdot |\sat(L)|}{|\sat(F_1)|} < \frac{(2/3)3^h \cdot 3^l}{3^f} = \frac{2}{3}$, a contradiction.
\end{proof}

\section{Proof of Lemma~\ref{lemma:technical-lemma}}\label{section:technical-lemma}

For the proof, we distinguish random formulas and random assignments from fixed formulas and fixed assignments by use of the tilde symbol. I.e., random formulas and assignments are denoted by $\tilde{F}$, $\tilde{G}$, $\tilde{H}$, $\tilde{\alpha}$, etc. The purpose of the subsection is to prove that, when $\tfrac{1}{2} < \delta < 1$ and $k = n^{1/3}$ (supposed integer), it holds that
\begin{equation}\label{eq:target}
\lim\limits_{n \rightarrow \infty} \Pr\nolimits_{\substack{\tilde F \sim \calH_2(n,m)}}\left[\theta\left(\tilde{F}^{\leq k},\tilde{F}^{> k}\right) < \tfrac{2}{3}\right] = 0
\end{equation}
The proof leans on a result of~\cite{rasmussen2025fixing,basse2025regularity} which implies that if we have a formula from $\tilde{G}\sim \calF_2(n,\delta n)$ and that we take an independent random assignment of size negligible compared to $\sqrt{n}$, i.e., $\tilde \alpha \sim \calF_1(n,o(\sqrt{n}))$, then $\tilde \alpha$ is likely to have an extension satisfying $\tilde{G}$ (as $n$ increases). 

\begin{theorem}[\protect{\cite[Theorem B.5]{rasmussen2025fixing}}]\label{theorem:saviour_thm}
Let $\tilde{G} \sim \calF_2(n,\delta n)$ and an independent $\tilde \alpha \sim \calF_1(n,2k)$ independent of $\tilde{G}$. If  $\delta < 1$ and $k = o(\sqrt{n})$, then $\lim_{n \rightarrow \infty} \Pr_{\tilde{G},\tilde{\alpha}}[\tilde{\alpha} \land \tilde{G} \in \SAT] = 1$.
\end{theorem}

\begin{proof}[Proof of Lemma~\ref{lemma:technical-lemma}] Since $\tilde{F}$ is drawn from $\calH_2(n,m)$, $\tilde{F}^{\leq k}$ and $\tilde{F}^{> k}$ are independent, hence
\begin{equation}\label{eq:split-formula}
\Pr\nolimits_{\tilde{F} \sim \calH_2(n,\delta n)}\left[\theta\left(\tilde{F}^{\leq k},\tilde{F}^{> k}\right) < \tfrac{2}{3}\right] 
= 
\Pr\nolimits_{\substack{\tilde{H} \sim \calH_2(n,k) \\ \tilde{G} \sim \calH_2(n,\delta n - k)}}\left[\theta\left(\tilde{H},\tilde{G}\right) <  \tfrac{2}{3}\right]
\end{equation}
where $\tilde{G}$ and $\tilde{H}$ are independent. Let $m = \delta n - k$ and $\tilde\alpha$ be a random assignment to $x_1,\dots,x_n$ of size $2k$, independent of $\tilde G$. We claim that the probability that $\tilde\alpha \land \tilde G$ is in $\SAT$ is converging to $1$ as $n$ increases
\begin{claim}\label{claim:1} The following holds
\begin{equation}\label{eq:random_assgn_sat}
\lim\limits_{n \rightarrow \infty}\Pr_{\substack{\tilde G \sim \calH_2(n,m) \\ \tilde\alpha \sim \calF_1(n,2k)}}\left[\tilde\alpha \land \tilde G\in \SAT\right] = 1.
\end{equation}
\end{claim}
\begin{proof} Theorem~\ref{theorem:saviour_thm}  gives $\lim_{n\to\infty}\Pr_{\substack{\tilde G' \sim \calF_2(n,\delta n) \\ \tilde\alpha \sim \calF_1(n,2k)}}
[\tilde\alpha \land \tilde G' \in\SAT]=1$. We have
\[
\Pr_{\substack{\tilde G' \sim \calF_2(n,\delta n) \\ \tilde\alpha \sim \calF_1(n,2k)}}
[\tilde\alpha \land \tilde G' \in\SAT] 
\leq
\Pr_{\substack{\tilde G'' \sim \calH_2(n,\delta n) \\ \tilde\alpha \sim \calF_1(n,2k)}}
[\tilde\alpha \land \tilde G'' \in\SAT]
\leq
\Pr_{\substack{\tilde G \sim \calH_2(n,m) \\ \tilde\alpha \sim \calF_1(n,2k)}}
[\tilde\alpha \land \tilde G \in\SAT]
\]
where the first inequality holds because $\calH_2(n,\delta n)$ may produce duplicate 
clauses and hence has fewer distinct constraints than $\calF_2(n,\delta n)$ on average, 
and the second holds because $m = \delta n -  k < \delta n$ and fewer clauses are easier to satisfy. The claimed statement follows.
\end{proof}
Fix $\varepsilon > 0, n \in \mathbb{N}$ and let $\scrP_{\varepsilon,n}$ be the set of formulas such that, 
$F \in \scrP_{\varepsilon,n}$ if and only if $\Pr_{\tilde\alpha}[\tilde \alpha \land F \in \SAT] \geq 1 - \varepsilon$ where $\tilde \alpha\sim \calF_1(n,2k)$.

\begin{claim}\label{claim:2}
For $\varepsilon > 0$ fixed, 
\begin{equation}
\lim\limits_{n \rightarrow \infty}
\Pr_{\tilde G \sim \calH_2(n,m)}\left[\tilde G \in \scrP_{\varepsilon,n}\right] = 1.
\end{equation}
\end{claim}
\begin{proof} We start from (\ref{eq:random_assgn_sat}). Since $\tilde{G}$ and $\tilde{\alpha}$ are independent, by the law of total expectation
applied to the indicator $\mathbbm{1}[\tilde{\alpha} \wedge \tilde{G} \in \mathrm{SAT}]$,
\begin{align*}
\mathbb{E}_{\substack{\tilde{G} \sim \calH_2(n,m) \\ \tilde{\alpha} \sim F_1(n,2k)}}\left[\mathbbm{1}[\tilde{\alpha} \wedge \tilde{G} \in \mathrm{SAT}]\right]
&= \mathbb{E}_{\tilde{G} \sim \calH_2(n,m) }\left[\mathbb{E}_{\tilde{\alpha} \sim F_1(n,2k)}\left[\mathbbm{1}[\tilde{\alpha} \wedge \tilde{G} 
\in \mathrm{SAT}] \mid \tilde{G}\right]\right]\\
&= \mathbb{E}_{\tilde{G} \sim \calH_2(n,m) }\left[\mathbb{E}_{ \tilde{\alpha} \sim F_1(n,2k)}\left[\mathbbm{1}[\tilde{\alpha} \wedge \tilde{G} 
\in \mathrm{SAT}]\right]\right],
\end{align*}
where the second equality uses independence of $\tilde{G}$ and $\tilde{\alpha}$.
Therefore, this expectation converges to $1$. Since the summand is in $[0,1]$,
Markov's inequality gives the following:
\[\Pr_{\tilde{G}}\left[1 - \mathbb{E}_{\tilde{\alpha}}\left[\mathbbm{1}[\tilde{\alpha} \wedge \tilde{G} 
\in \mathrm{SAT}]\right] > \varepsilon\right] 
\leq \frac{\mathbb{E}_{\tilde{G}}\left[1 - \mathbb{E}_{\tilde{\alpha}}\left[\mathbbm{1}[\tilde{\alpha} 
\wedge \tilde{G} \in \mathrm{SAT}]\right]\right]}{\varepsilon} \xrightarrow{n \rightarrow \infty} 0.
\]
It follows by definition of $\scrP_{\varepsilon,n}$ that
\[
\Pr_{\tilde{G}}\left[\tilde{G} \notin \scrP_{\varepsilon,n}\right]
= \Pr_{\tilde{G}}\left[\mathbb{E}_{\tilde{\alpha}}\left[\mathbbm{1}[\tilde{\alpha} \wedge 
\tilde{G} \in \mathrm{SAT}]\right]\right] \xrightarrow{n \rightarrow \infty} 0.\qedhere
\]
\end{proof}
Now, let us describe an alternative construction of $\tilde\alpha$. 

\begin{itemize}
\item Draw a matching formula $\tilde Z$ uniformly from $MF(n,k)$ independently of $\tilde G$.
\item Draw $\tilde \alpha_{\tilde{Z}}$ uniformly at random in $\sat(\tilde{Z})$; 
\end{itemize}

Intuitively, since $\tilde Z$ is independent of $\tilde G$,  $\tilde \alpha_{\tilde{Z}}$ is seen by $\tilde G$ as an independent random assignment of size $2k$ over $x_1,\dots,x_n$, as if it came from $\calF_1(n,2k)$. 

Since the variables $x_1,\dots,x_n$  play a symmetric role in the construction of $\tilde{Z}$, any two fixed assignments $\alpha$ and $\beta$ of size $2k$ over $x_1,\dots,x_n$ are equally likely to be solutions of $\tilde{Z}$. Formally,
\begin{equation}
\Pr_{\tilde{Z}}\left[\alpha \in \sat(\tilde{Z})\right] = \Pr_{\tilde{Z}}\left[\beta \in \sat(\tilde{Z}) \right]
\end{equation}
Since $\tilde{Z}$ is in $MF(n,k)$ we have $\Ex[|\sat(\tilde{Z})|] = 3^{k}$. So the above probability is $3^{k}$ divided by the number of possible assignments, i.e., $3^{k}\left(2^{2k}\binom{n}{2k}\right)^{-1}$. Thus,
\begin{equation}
\Pr_{\tilde{Z}}\left[\tilde\alpha_{\tilde Z} = \alpha\right] = \tfrac{1}{3^{k}} \Pr_{\tilde{Z}}\left[\alpha \in \sat(\tilde{Z})\right]
= {\textstyle \left(2^{2k}\binom{n}{2k}\right)^{-1}} = \Pr_{\tilde{\alpha} \sim \calF_1(n,2k)}\left[\tilde{\alpha} = \alpha\right].
\end{equation}
It follows that, for $F$ a fixed CNF, 
\begin{equation}
\Pr\limits_{\tilde Z}\left[ \tilde\alpha_{\tilde Z} \land F \in \SAT \right] = \Pr\limits_{\tilde{\alpha} \sim  \calF_1(n,2k)}\left[\tilde\alpha \land F \in \SAT\right].
\end{equation}
Let us define $\scrP_F$ as follows: $F' \in \scrP_F$ if and only if $\theta(F',F) \geq 2/3$. We have that $\Pr_{\tilde{Z}}\left[ \tilde\alpha_{\tilde Z} \land F \in \SAT \right]$ equals
\begin{align*}
& 
\Pr_{\tilde{Z}}\left[\tilde{Z} \not\in \scrP_F\right] \Pr_{\tilde{Z}}\left[\tilde{\alpha}_{\tilde Z} \land F \in \SAT \mid \tilde{Z} \not\in \scrP_F\right]
+
\Pr_{\tilde{Z}}\left[\tilde{Z} \in \scrP_F\right]\Pr_{\tilde{Z}}\left[\tilde{\alpha}_{\tilde Z} \land F \in \SAT \mid \tilde{Z} \in \scrP_F\right]
\\
&\leq
\Pr_{\tilde{Z}}\left[\tilde{Z} \not\in \scrP_F\right]\Pr_{\tilde{Z}}\left[\tilde{\alpha}_{\tilde Z} \land F \in \SAT \mid \tilde{Z} \not\in \scrP_F\right] + \Pr_{\tilde{Z}}\left[\tilde{Z} \in \scrP_F\right]
\\
&\leq
\Pr_{\tilde{Z}}\left[\tilde{Z} \not\in \scrP_F\right]\cdot\tfrac{2}{3} + \Pr_{\tilde{Z}}\left[\tilde{Z} \in \scrP_F\right] 
= \tfrac{2}{3} + \tfrac{1}{3}\cdot\Pr_{\tilde{Z}}\left[\tilde{Z} \in \scrP_F\right] 
 = \tfrac{2}{3} + \tfrac{1}{3}\cdot\Pr_{\tilde{Z}}\left[\theta(\tilde{Z},F) \geq \tfrac{2}{3}\right].
\end{align*}
In particular, when $F \in \scrP_{\varepsilon,n}$ we also have $\Pr_{\tilde{Z}}[ \tilde\alpha_{\tilde Z} \land F \in \SAT ] \geq 1 - \varepsilon$, and therefore 
\begin{equation}
F \in \scrP_{\varepsilon,n} \quad \Rightarrow \quad \Pr_{\tilde{Z}}\left[\theta(\tilde{Z},F) \geq \tfrac{2}{3}\right] \geq 1 - 3\varepsilon.
\end{equation}
When $\tilde H$ is conditioned on being a matching formula it is equally likely to be any formula from $MF(n,k)$, so it behaves like $\tilde{Z}$. 
\begin{equation}\label{eq:nice-bound}
\Pr_{\tilde{H},\tilde{G}}\left[\theta(\tilde{H},\tilde{G}) \geq \tfrac{2}{3} \mid \tilde{G} \in \scrP_{\varepsilon,n}, \tilde{H} \in \MF(n,k)\right] \geq 1 - 3\varepsilon.
\end{equation}
Since $\tilde{H}$ is almost always a matching formula (Lemma~\ref{lemma:the_first_k_clauses_in_H_are_a_matching}) and since $\tilde{G}$ is almost always in $\scrP_{\varepsilon,n}$ (Claim~\ref{claim:2}), we have that for every fixed $\varepsilon > 0$ we have that 
\begin{multline*}
\hfill \Pr\nolimits_{\substack{\tilde{H} \sim \calH_2(n,k) \\ \tilde{G} \sim \calH_2(n,\delta n - k)}}\left[\theta\left(\tilde{H},\tilde{G}\right) < \tfrac{2}{3}\right] = 
\Pr_{\tilde{H},\tilde{G}}\left[\theta\left(\tilde{H},\tilde{G}\right) < \tfrac{2}{3}\mid \tilde{G} \in \scrP_{\varepsilon,n},\tilde{H} \in \MF(n,k)\right]
\\
\hfill \cdot \Pr_{\tilde{H},\tilde{G}}\left[\tilde{G} \in \scrP_{\varepsilon,n},  \tilde{H} \in \MF(n,k)\right] + o(1)
\\
\hfill \leq \Pr_{\tilde{H},\tilde{G}}\left[\theta\left(\tilde{H},\tilde{G}\right) < \tfrac{2}{3}\mid \tilde{G} \in \scrP_{\varepsilon,n},\tilde{H} \in \MF(n,k)\right] + o(1) 
\leq 3\varepsilon + o(1).
\end{multline*}
This holds for every $\varepsilon > 0$, so we finally get (\ref{eq:target}), and the proof of Lemma~\ref{lemma:technical-lemma} is done.
\end{proof}

One last comment about the proof. A careful reader could object that, in terms of graph parameters, Theorem~\ref{theorem:compilability_threshold_monotone} requires the treewidth to be large and the degree to be small, while we never discuss the degree for Theorem~\ref{theorem:compilability_thresholds}, which may seem odd. And indeed, the fact that the maximum degree of $G_F$ is generally small (less than $\log(n)$) with high probability is likely important for Theorem~\ref{theorem:compilability_thresholds}, but, in a sense, it is already baked into Theorem~\ref{theorem:saviour_thm}, which is why we never explicitly need it. Theorem~\ref{theorem:saviour_thm} works in the sparse CNF regime only (i.e., $\calF_2(n,\delta n)$ with $\delta > 0$), which, in itself, forces small degree with high probability. 

\section{Discussion and Future Work}

Several natural variations of Theorem~\ref{theorem:compilability_thresholds} need to be investigated. First, we can modify the \emph{initial language}; for instance, instead of sparse $2$-CNF one may consider $k$-CNF for $k > 2$, or even $k$-XORSAT or $k$-NAESAT formulas, sparse or dense. Second, the \emph{target language} can also be altered. One direction is to constrain it more, for instance, by saying that the variable order of the OBDD is fixed in advance. In that case, we conjecture that the OBDD-size is large even below the treewidth threshold, as there is no way to exploit the structure of the primal graph through the variable order. In the opposite direction, there are many compilation languages that generalize and are more \emph{succinct} than OBDD~\cite{DarwicheM02}, which could be considered as alternative target languages. We plan to show that Theorem~\ref{theorem:compilability_thresholds} also holds for compilation to \emph{structured DNNF} circuits, and we have, in a sense, paved the way for the proof by insisting on working with treewidth rather than pathwidth. We also believe that replacing sparse $2$-CNF with sparse $3$-CNF is feasible. Currently, the exact location of the satisfiability threshold for random $3$-CNF is unknown.  In addition, to our knowledge, the satisfiability threshold is not the same as the treewidth threshold (which we are confident exists), but the result can perhaps be stated and proved using placeholder symbols instead of the actual values.

We also want to research alternative proofs of  Theorem~\ref{theorem:compilability_thresholds}. There could be a clever proof relying on the fact that, as $n$ increases, almost all 2-CNF formulas represent \emph{unate functions}, i.e., functions that are monotone modulo consistent literal renaming~\cite{Allen2007}. Note that this is \underline{not} the same as saying that almost all 2-CNF are monotone modulo consistent literal renaming. Nevertheless, knowing Theorem~\ref{theorem:compilability_threshold_monotone} for monotone CNFs, there is some hope for a proof of Theorem~\ref{theorem:compilability_thresholds} that would try to reduce the CNF to one that is monotone. But this requires showing that almost all 2-CNF formulas from $\calF_2(n,\delta n)$, \emph{for any fixed $\delta$}, represent unate functions as $n$ increases. This does not derive directly from~\cite{Allen2007}, but sounds plausible.  

Lastly, as a follow-up to~\cite{GuptaRM20}, empirical evaluations would be insightful. Setting up these experiments is not exactly easy since finding the variable order that gives the smallest OBDD is \NP-hard~\cite{BolligW96}, but that would allow us to see whether phase transitions are actually visible for reasonable $n$ at the thresholds given in this paper.

\bibliography{main}

\section{Appendix -- Proofs}

\twThresholdForCNF*
\begin{proof}
When $F \sim \calF_2(n,\delta n)$ we are likely to have less than $\delta n$ edges in $G_F$ because we can select up to four clauses that correspond to the same edge $\{x,y\}$ namely, $(x \lor y)$, $(\neg x \lor y)$, $(x \lor \neg y)$ and $(\neg x \lor \neg y)$. Let $E_n$ be all possible edges over $X_n = \{x_1,\dots,x_n\}$. Consider the graph $G'_F$ obtained by randomly adding $\delta n - |E(G_F)|$ distinct edges of $E_n \setminus E(G_F)$ to $G_F$. We have $G'_F \sim \calGr(n,\delta n)$. 

Since $\tw(G_F) \leq \tw(G'_F)$ we immediately get by Theorem~\ref{theorem:tw_threshold_for_graphs_uniform},  that, when $\delta < 1/2$, $\lim_{n \rightarrow \infty} \Pr[\tw(G_F) \geq 3] \leq \lim_{n \rightarrow \infty} \Pr[\tw(G'_F) \geq 3] = 0$.

Now, when $\delta > 1/2$, Theorem~\ref{theorem:tw_threshold_for_graphs_uniform}, gives us $\lim_{n \rightarrow \infty} \Pr[\tw(G'_F) \geq cn] = 0$. We say that the edge $\{x,y\}$ has weight $k \in \{0,1,2,3,4\}$ in $G_F$ if $k$ of the clauses listed above are in $F$. We claim that, with high probability, $G_F$ does not contain more than $\log(n)$ edges with weight $2$ or more. It will follow that $|E(G_F)| \geq \delta n - 4\log(n)$ with high probability. Furthermore, since adding an edge to a graph cannot increase the treewidth by more than $2$, we will have $\tw(G'_F) \geq \tw(G_F) + 8\log(n)$ with high probability. But then $
\lim_{n \rightarrow \infty} \Pr[\tw(G_F) \geq cn + 8\log(n)] = 0$ will follow, and therefore $
\lim_{n \rightarrow \infty} \Pr[\tw(G_F) \geq 2cn] = 0$.

It remains to be proved that the number of edges with weight at least $2$ is small. Let $N = 4\binom{n}{2}$, the probability for edge $e := \{x,y\}$ to have weight $0$ is $\binom{N-4}{m}/\binom{N}{m}$, the probability for $e$ to have weight $1$ is $4\binom{N-4}{m-1}/\binom{N}{m}$. Thus the probability for $e$ to have weight $\geq 2$ is 
\begin{multline*}
1 - \frac{(N-4)!(N-m)!}{(N-4-m)!N!} - \frac{4m(N-4)!(N-m)!}{(N-3-m)!N!}
\\
\hfill = 1 -  \left(1 - \frac{m}{N}\right)\left(1 - \frac{m}{N-1}\right)\left(1 - \frac{m}{N-2}\right)\left(1 - \frac{m}{N-3}\right) 
\\
\hfill - \frac{4m}{N-3} \left(1-\frac{m}{N}\right)\left(1 - \frac{m}{N-1}\right)\left(1 - \frac{m}{N-2}\right)
\\
\hfill  \leq 1 - \left(1 - \frac{m}{N-3}\right)^4  - \frac{4m}{N-3} \left(1 - \frac{m}{N-3}\right)^3 
\\
\hfill  = 1 - \left(1 - \frac{4m}{N-3} + O\left(\frac{m^2}{N^2}\right) \right) - \frac{4m}{N-3} \left(1 + O\left(\frac{m}{N}\right)\right)^3 
\\
\hfill = O\left(\frac{m^2}{N^2}\right) 
\end{multline*}
So the expected number of edges with weight $\geq 2$ is at most $N/4\cdot O(m^2/N^2) = O(m^2/N) = o(1)$. Thus, by Markov's bound, the probability that more than $\log(n)$ edges have weight $\geq 2$ goes to $0$ as $n$ increases.
\end{proof}

\smallDegree*
\begin{proof}
We consider $G \sim \calGr(n,p = 2\delta/n)$, show that $\lim_{n \rightarrow \infty}\Pr[\Delta(G) \geq \log(n)] = 0$ and use~\cite[Corollary 1.16]{JansonLR00} to conclude that $\lim_{n \rightarrow \infty}\Pr[\Delta(G) \geq \log(n)] = 0$ when $G \sim \calGr(n,m=\delta n)$. For a vertex $v \in V_n = \{v_1,\dots,v_n\}$, let $E_v = \{ \{v,v'\} \mid v' \in V_n \setminus \{v\}\}$ be all possible edges incident to $v$ and let $\deg(v)$ be the degree of $v$ in $G$. Let $\calE = \{E \subseteq E_v \mid |E| = \log(n)\}$ be the collection of all sets of $\log(n)$ edges incident to $v$.
\[
\Pr[\deg(v) \geq \log(n)] \leq \Pr\bigg[ \bigcup_{E \in \calE} E \subseteq E(G)\bigg] \leq \sum_{E \in \calE} \Pr\bigg[ \bigcap_{e \in E} e \in E(G) \bigg] = |\calE| \left(\frac{\delta  n}{\binom{n}{2}}\right)^{\log(n)}
\]
We recall that $(\frac{a}{b})^b \leq \binom{a}{b} \leq (\frac{ea}{b})^b$ holds for every $b \leq a$. We have $|\calE| = \binom{n-1}{\log(n)} \leq \left(\frac{en}{\log(n)}\right)^{\log(n)}$ so $\Pr[\deg(v) \geq \log(n)] \leq \left(\frac{4e\delta}{\log(n)}\right)^{\log(n)}$. Hence
\[
\Pr[\Delta(G) \geq \log(n)] \leq \sum_{v \in V_n} \Pr[\deg(v) \geq \log(n)] \leq n  \left(\frac{4e\delta}{\log(n)}\right)^{\log(n)} \xrightarrow{n \rightarrow \infty} 0 \qedhere
\]
\end{proof}

\fromHtoF*
\begin{proof}
By application of Bayes rule.
\begin{align*}
\Pr\limits_{F \sim \calF_2(n,m)}&[F \in \mathscr{P}] = \Pr\limits_{F \sim \calH_2(n,m)}[F \in \mathscr{P} \mid F \text{ is simple}] 
\\&= \frac{\Pr_{F \sim \calH_2(n,m)}[F \in \mathscr{P} \land F \text{ is simple}]}{\Pr_{F \sim \calH_2(n,m)}[F \text{ is simple}]}
\leq \frac{\Pr_{F \sim \calH_2(n,m)}[F \in \mathscr{P}]}{\Pr_{F \sim \calH_2(n,m)}[F \text{ is simple}]}
\end{align*}
Since the denominator is at least $(1-o(1))e^{-\delta^2 - \delta}$ (by (\ref{eq:proba-simple})) the fraction converges to $0$ when its numerator converges to $0$ as $n$ increases. 
\end{proof}

\notTooManyNonUniqueClauses*
\begin{proof}
By Markov bound, using (\ref{eq:birthday}), $\Pr[F \text{ has} \geq \sqrt{n} \text{ non-unique clauses}]$ is at most
$$
\delta \sqrt{n} \left(1 - \left(1-\frac{1}{n^2}\right)^{\delta n}\right) \leq  \delta \sqrt{n} \left(1 - \left(1  - \frac{\delta}{n} + O\left(\frac{1}{n^2}\right)\right)\right) \leq  \delta \sqrt{n} \left(\frac{\delta}{n} + O\left(\frac{1}{n^2}\right)\right),
$$ 
where we have used that the Taylor expansion of $(1-x^2)^{\delta/x}$ at $0$ is $1 - \delta x +O(x^2)$ (the function and its derivative are not defined at $0$ but are analytically extendable to $0$). The right-hand side function converges to $0$ has $n$ goes to infinity.
\end{proof}

\largeTWinH*
\begin{proof}
Let $E$ be all possible edges over $X_n$. We call $G$ the graph obtained from $G_F$ by adding $\delta n - |E(G_F)|$ distinct edges chosen uniformly at random from $E \setminus E(G_F)$. Observe that $G \sim \calGr(n,\delta n)$ so, by Theorem~\ref{theorem:compilability_threshold_monotone}, $\lim_{n \rightarrow \infty} \Pr[\tw(G) \geq 
d n]  = 1$ for some $d > 0$.  

How many edges are added to get from $G_F$ to $G$? By Lemma~\ref{lemma:not-too-many-non-unique-clauses-in-H} there are at least $\delta n - \sqrt{n}$ unique clauses in $F$ and, reasoning as in the proof of Lemma~\ref{lemma:tw_threshold_for_2CNF}, with high probability less than $\log(n)$ edges of $G_F$ have weight $2$ or more. So $|E(G_F)| \geq \delta n -\Omega(\sqrt{n})$ with high probability. Thus with high probability we add $O(\sqrt{n})$ and therefore $\tw(G_F) \geq \tw(G) -\Omega(\sqrt{n})$. It follows that $\lim_{n \rightarrow \infty} \Pr[\tw(G_F) \geq dn -\Omega(\sqrt{n}) \geq
d n/2]  = 1$.
\end{proof}

\largeMimwInH*
\begin{proof}
Follows from combining Lemma~\ref{lemma:large-tw-in-H} and Theorem~\ref{theorem:tw_mmw}.
\end{proof}

\largeMatchingFormula*
\begin{proof}
Let $L(T)$ be the set of leaves of $T$. If $t \in L(T)$, then $|V_t| = 1$ and $mm_G(V_t,\bar V_t) \in \{0,1\}$. Thus, $mmw(G,T) = \max_{t \in V(T)} mm_G(V_t,\bar V_t) \leq \max_{t \in V(T) \setminus L(T)} mm_G(V_t,\bar V_t) + 1$. So 
$$
|M(F)| = \max_{t \in V(T(F)) \setminus L(T(F))} mm_{G_F}(V_t,\bar V_t)\geq \min_T\max_{t \in V(T) \setminus L(T)} mm_{G_F}(V_t,\bar V_t) \geq  mmw(G_F) - 1
$$ with $T$ ranging over all binary tree decompositions of $G_F$. By Lemma~\ref{lemma:large-mimw-in-H}, it follows that $\lim_{n \rightarrow \infty} \Pr[|M(F)| \geq \gamma n] = 1$. 
\end{proof}

\firstKclausesAreMatching*
\begin{proof}
    Let $F^{\leq k} = C_1 \land \dots \land C_k$. Fix $i \in [k]$. Let $X$ be the number of pairs $(C_i,C_j)$, $i \neq j$, such that $\var(C_i) \cap \var(C_j) \neq \emptyset$. Once $C_i$ is chosen, there are $8(n - 1) - 4$ possible clauses that share a variable with it. So
    \[\Pr[C_i \text{ and } C_j \text{ share a variable}] = \frac{8(n-1)-4}{4\binom{n}{2}} \leq \frac{4}{n}\]
    Thus, $\Ex[X] \leq \binom{k}{2} \frac{4}{n} \leq O(\frac{k^2}{n})$. Since $k=o(\sqrt{n})$, we have $\lim_{n\rightarrow \infty}\Ex[X] = 0$. So, by Markov bound, $\lim_{n\rightarrow \infty}\Pr[X \geq 1] \leq \lim_{n\rightarrow \infty}\Ex[X] = 0$.
\end{proof}

\matchingOnFirstKclauses*
\begin{proof}
Let $clauses(F)$ be the multiset of clauses of $F$ (with repetition). Let $\calC$ be the set of multisets of $\delta n$ clauses such that, if $clauses(F) \in \calC$, then $|M(F)| \geq \gamma n$, with $\gamma$ the constant of Corollary~\ref{corollary:large-matching}. By Corollary~\ref{corollary:large-matching}, for any $t$ 
\begin{equation}\label{eq:sum_over_good_clause_multisets}
\Pr[|M(F,k)| \geq t] = \sum_{S \in \calC} \Pr[|M(F,k)| \geq t \mid clauses(F) = S] \Pr[clauses(F) = S] + o(1)
\end{equation}
Formulas from $\calH_2(n,\delta n)$ that have the same clause multiset $S = \{C_1,C_2,\dots,C_{\delta n}\}$ also have the same matching $M(F)$. Let $C_{i_1},\dots,C_{i_s}$ be the clauses  of $M(F)$ for a fixed $S$. We ask the probability that a formula with clause multiset $S$ has $t$ clauses from $M(F)$ in its first $k$ clauses. This probability is upper bounded by the case where $C_i \neq C_j$ for every $i \neq j$, since having a duplicate of $C_{i_1}$ would increase the probability that one is in the first $k$ clauses. So we assume $F$ is simple and bound 
$$
\Pr[|M(F,k)| < t \mid clauses(F) = S]
$$
from above. Let $\sigma$ be a random permutation of $[\delta n]$ and $\sigma(F) := C_{\sigma(1)} \land \dots \land C_{\sigma(\delta n)}$. Each $\sigma(F)$ is equally likely in $\calH_2(n,\delta n)$. Our probability is then that $\sigma(F)^{\leq k}$ intesects $M(F)$ on $t$ clauses, i.e., that $\{\sigma(1),\dots,\sigma(k)\} \cap \{i_1,\dots,i_s\}$ has size at least $t$. This is exactly the probability that a random variable $X$ following the hypergeometric distribution of parameters $(\delta n, s,k)$ is at least $t$ (at least $t$ red balls when picking $k$ balls from a bin of $\delta n$ balls, of which $s$ are red). Now suppose $s$ is at least $\gamma n$, then the expected value of $X$ is at least $\gamma k/\delta$. Thus, if $t = \gamma k/(2\delta)$, then the tail bound for the hypergeometric distribution gives us 
$$
\Pr[X < \tfrac{\gamma k}{2\delta}] \leq \Pr[X - \Ex[X] < -\tfrac{\gamma k}{2\delta}] \leq \exp(-\tfrac{\gamma^2 k}{2\delta^2}) = \exp(-\tfrac{\gamma^2 n^{1/3}}{2\delta^2})
$$
So, when $s \geq \gamma n$, $\Pr[|M(F,k)| < \tfrac{\gamma k}{2\delta} \mid clauses(F) = S] \leq \Pr[X < \tfrac{\gamma k}{2\delta}] = o(1)$ when $n \rightarrow \infty$. Plugging this into (\ref{eq:sum_over_good_clause_multisets}) yields
$$
\Pr\left[|M(F,k)| <  \tfrac{\gamma k}{2\delta}\right] = o(1)  \quad \text{ when } \quad n \rightarrow \infty.
$$
\end{proof}

\end{document}